\documentclass{article}
\usepackage[utf8]{inputenc}
\usepackage{xcolor}
\usepackage{xcolor}
\usepackage{lineno} 

\usepackage{graphicx,multicol}

\usepackage{epic,eepic,epsfig,comment}

\usepackage{amssymb}
\usepackage[square,numbers]{natbib}
\usepackage{amsmath} %new

\usepackage{amsthm} 
\usepackage{cases}
\usepackage{authblk}
\usepackage{tikz}
\usepackage[upright]{fourier}
\usepackage{tkz-graph}
\usetikzlibrary{arrows}
\usetikzlibrary{calc}
\usetikzlibrary{shapes}
\usetikzlibrary{decorations.pathreplacing}
\usetikzlibrary {fit,shapes.geometric} 

\newcommand{\rev}[1]{\textcolor{black}{#1}}

\definecolor{forestgreen}{rgb}{0.13, 0.55, 0.13}

\newtheorem{theorem}{Theorem}[section]

\newtheorem{lemma}[theorem]{Lemma}

\newtheorem{observation}[theorem]{Observation}

\newtheorem{corollary}[theorem]{Corollary}

\newtheorem{claim}{Claim}

\title{Complexity results for a cops and robber game on directed graphs
}

\author[1]{Walid Ben-Ameur}
\author[1]{Alessandro Maddaloni }%\footnote{corresponding author: @telecom-sudparis.eu}}

\affil[1]{SAMOVAR, Télécom SudParis, Institut Polytechnique de Paris, Palaiseau, France}

\tikzset{loop/.style={min distance=5mm,looseness=5}}
\tikzset{tinoeud/.style={draw, circle, minimum height=0.01cm}}

\begin{document}

\maketitle

\begin{abstract}
%\wal{A marginal comment: in introduction, when we cite papers, sometimes we put the names of the authors, and sometimes not...}\\
{We investigate a cops and robber game on directed graphs, where the robber moves along the arcs of the graph, while the cops can select any position at each time step. Our main focus is on the cop number: the minimum number of cops required to guarantee the capture of the robber. We {prove} that 
%determining
{deciding} whether the cop number {of a digraph} is {equal to} 1 
%for general digraphs 
is NP-hard, whereas this is decidable in polynomial time for tournaments. Furthermore, we show that computing the cop number for general digraphs is fixed parameter tractable when parameterized by a generalization of vertex cover. However, for tournaments, tractability is achieved with respect to the minimum size of a feedback vertex set. Among our findings, we prove that the cop number of a digraph is equal to that of its reverse digraph, and we draw connections to the matrix mortality problem.}
%We investigate a cops and robber game on directed graphs where robber moves along graph arcs while cops can select any position at each time step. We mainly focus on the cop number: the minimum number of cops required to guarantee the capture of the robber.  
%We show that it is   NP-hard to check whether the cop number is $1$ for general digraphs, while this is easy to do in the case of tournaments. Computing the cop number on general digraphs is FPT parameterized by some vertex cover generalization, while tractability with respect to the minimum size of a feedback vertex set is achieved for tournaments.  
%Among our results, we show that the cop number of a digraph and its reverse digraph are equal, and we highlight some connections with the matrix mortality problem.  
\end{abstract}

{\bf{Keywords:}}
Cops and robber games, Cop number, Complexity, Fixed-parameter tractability, Digraphs, Hunters and rabbit, Matrix mortality.

\section{Introduction}

%\section{Game definitions}

\label{sec:def}
Consider a directed graph  $D = (V,A)$ where loops are allowed. 
%and $ \delta^+(D)$, the minimum outdegree of $D$, is at least one. 
A game, introduced in \cite{ourpaper}, and refereed as HCISR (helicopter cops invisible slow robber) is defined as follows. A set of cops wants to capture a robber moving on $D$.  At step 1, the cops pick $W_1\subseteq V$, then the robber picks a vertex $r_i\in V$.  At step $i+1$ (for any $i\ge 1$), the cops pick $W_{i+1}\subseteq V$ and the robber picks $r_{i+1} \in N^+(r_{i})$, where $N^+(v)$  denotes the set of vertices $y$ such that \rev{$(v,y)$} is an arc of $D$.
In other words, at each step the cops can pick any vertex, while the robber must pick a vertex adjacent from his current position. 
We say that the cops capture the robber if either $r_i \in W_i$ or $N^+(r_{i-1}) = \emptyset$   for some $i\ge 1$.
The capture is at time $t$, if $t$ is the minimum index such that either $r_t \in W_t$ or $N^+(r_{t-1}) = \emptyset$.
Note that the cops do not know the vertex picked by the robber; their strategy is defined by the sequence $(W_i)_{i\ge 1}$, while the robber strategy is defined by $(r_i)_{i\ge 1}$. 

A cop strategy is winning if it allows to capture  the robber regardless of his strategy. We say that a cop strategy uses $h$ cops when $h$ is the maximum of $|W_i|$ over all {time} indices. 
The capture time of a cop winning strategy is the maximum time step $\Omega$, over all possible robber's strategies, such that the cops capture the robber at time $\Omega$. \\
Let $R_i$ be the set of vertices where the robber can be at step $i$. 
Observe that if the robber was not yet captured at time step $i- 1$, then we have $R_i = N^+(R_{i-1}) \setminus W_i$, where  $N^+(S):=\{ y \in V \ | \ \exists x \in S, \mbox{ with } (x,y) \in A$\}) for $S \subset V$. 
Capturing  the robber at time step $i$ is equivalent to have $R_i = \emptyset$ and $R_{i-1} \neq \emptyset$. 

The cop number of $D$, denoted by \rev{$\mbox{cn}(D)$}, is the minimum $h$ such that there exists a winning strategy using $h$ cops.  If $\mbox{cn}(D) \leq k$, we say that $D$ is a $k$-copwin digraph. The capture time using $\mbox{cn}(D)$ cops and applying the best strategy (the one minimizing capture time) is denoted by $ct(D)$. 
%The capture attempts number of $D$, denoted by \rev{$ca(D)$}  is the minimum $c$ such that there exists a cop winning strategy using $c$ attempts. The capture time using $c\ge ca(D)$ attempts, denoted by $ct(D,c)$ is the minimum capture time of a cop winning strategy using $c$ attempts. 
%When fixing a time limit $l$, we can define $\mbox{cn}(D,l)$ (resp. \rev{$ca(D,l)$}) as the minimum number of cops (resp. capture attempts)  needed for the cops to capture no later than time step $l$. \rev{Observe that $ca(D,l) \le c$ is equivalent to $ct(D,c) \le l$.}
 %As an example consider a directed cycle $Y$ on $n$ vertices: we have $\mbox{cn}(Y)=1, \ ca(Y)=n, \ ct(Y,n)=1, \ \mbox{cn}(Y,l)=\lceil \frac{n}{l} \rceil, \ ca(Y,l)=n$. % Notice that if the cycle contains loops on more than one vertex, the cop number raises to two. 
 %Notice that if the cycle contains a loop around each vertex, the cop number raises to $2$ while it equals $1$ even if there is  at least one vertex without loop.  
 {For illustration}, consider a digraph $D$ (from \cite{ourpaper}) containing a directed cycle on $4$ vertices numbered from $1$ to $4$ and $3$ loops (no loop around $4$). 
A winning strategy using only $1$ cop is given in the following table where the cop position $W_t$ and the robber territory $R_t$ are provided for $1 \le t \le 8$. The cop number of this graph is then equal to $1$ while the capture time equals $8$ (it is not possible to guarantee capture before time step $8$). 
\begin{table}[h!]%tbp]
\label{tab:cop}
\center
\hspace*{-0.3cm}
\begin{tabular}{|c|c|c|c|c|c|c|c|c|}
\hline
$t$  & 1 & 2 & 3& 4& 5 &6 & 7 & 8 \\
\hline
$W_t$  & $\{ 3 \}$& $\{1   \}$    & $\{2\}$  & $\{3 \}$& $\{1 \}$
& $\{2 \}$& $\{3 \}$ & $\{ 1\}$ \\
\hline
$R_t$ & $\{1, 2, 4 \}$ & $\{ 2, 3 \}$ & $\{3, 4 \}$ & $\{1, 4 \}$& $\{ 2 \}$& $\{3 \}$& $\{ 4 \}$ & $\emptyset$\\
\hline
\end{tabular}
\caption{{ Winning strategy for the cops on a cycle with four vertices and loops on all vertices except vertex $4$.}}
\end{table}
 
{One can also easily check that when}  $D$ is a path on $n$ vertices, then $\mbox{cn}(D) = 0$ and $ct(D) = n+1$. \\ \ \\
 The cop number is upper bounded by $1$ plus the directed pathwidth $pw(D)$ of the graph \cite{ourpaper}. If capture is required to be monotone (i.e., the robber territory never increases), then at least $pw(D)$ cops are needed while $1+pw(d)$ is still a valid upper bound.   The size of a minimum feedback vertex set is also an obvious upper bound of $\mbox{cn}(D)$ while a lower bound is given by $\max_{S \subset V} \delta^+(D[S])$ (i.e., the maximum through all induced subgraphs of the minimum outdegree). A polynomial-time algorithm is presented  in \cite{ourpaper}  to compute such a lower bound. Some links with the no-meet matroid \cite{BENAMEUR2022} and the meeting time of synchronous directed walks are also provided in \cite{ourpaper}.   

An undirected version of our  game  has been considered under the name hunters and rabbit (or prince and princess) \cite{ABRAMOVSKAYA201612,journals/combinatorics/BritnellW13}. %, where helicopter cops play against an invisible robber moving at speed one. 
In fact,  hunters and rabbit problem can be seen as a special case of our game where the digraph is symmetric and does not contain loops. 
In hunters and rabbit, the cop (or hunter) number can be computed for special graphs such as paths, cycles, grids or hypercubes \cite{ABRAMOVSKAYA201612,BOLKEMA2019360}. 
Graphs for which the cop number is equal to $1$ have been characterized in  \cite{journals/combinatorics/BritnellW13,HASLEGRAVE201412}.
An upper bound on the cop number for trees is given in \cite{gruslys2015catching}. 
Bounds related to the pathwidth of the undirected graph are also provided in 
\cite{dissaux_et_al:LIPIcs.MFCS.2023.42}. 
Note that the complexity of computing the cop number of the hunters and rabbit game is still not known, however the problem is FPT when parameterized by the size of a minimum vertex cover \cite{dissaux_et_al:LIPIcs.MFCS.2023.42}.

More generally, cops and robber games have been introduced  in \cite{Quillot} and intensively studied over the last four decades in their many variants. For a nice introduction, mainly on undirected graphs, see, e.g., \cite{Bontato}.   \\
In \cite{SEYMOUR199322} is introduced a variant with helicopter cops, namely the cops need not follow the graph edges. In their version the robber must follow the graph edges, but has unlimited speed so in one time step he (or she) can visit any vertex on a cop free path starting from his position. When the robber is visible to the cops, the cop number of an undirected graph, equals the treewidth of the graph plus one \cite{SEYMOUR199322}. If the robber is invisible to the cops, the cop number equals the pathwidth of the graph (or vertex separation number) plus one \cite{KIROUSIS1985, KINNERSLEY1992}. In both versions, the winning strategy for the cops is monotone, namely the robber territory never increases. These kind of games are also formulated as graph searching problems, for a bibliography see \cite{FOMIN2008}. \\
Much work has been done mostly on undirected graphs, but different game versions (and their relation to graph parameters) are studied on directed graphs too, see for example \cite{Barat,Berwanger,Janssen,10.5555,YANG20081822}. 
The invisible version of the game on directed graphs has been studied in \cite{Barat}  and the cop number is shown to be either the directed pathwidth $pw(D)$ of the digraph or $pw(D) +1$. \\
Capture time of a strategy using a fixed number of cops has been introduced in \cite{BONATO20095588}. In most of the literature, the concept is related to the (original) variant of the game where cops have no helicopter, so they both follow the edges with speed one and the robber is visible. In \cite{BONATO20095588} they showed that for graphs with cop number one, the capture time is at most $n-4$, where $n$ is the number of vertices of the graph. Later, it was proved in \cite{BRANDT2020} that the bound $O(n^{k+1})$ is tight for graphs whose cop number is $k\ge 2$. The latter result extends to directed graphs, while \cite{KINNERSLEY2018} showed that there are directed graphs with cop number one whose capture time is $\Omega(n^2)$.

\section{Paper's  contributions and organization}

{Our first results show that it is NP-hard to compute $\mbox{cn}(D)$ and it is also NP-hard to compute the minimum number of cops needed  to capture the robber no-later than some given time step $l$. These two results were already claimed by the authors in \cite{ourpaper}, but as observed by {Amarilli et al. }\cite{error}, the proposed reductions need corrections. \\} 
 %correct the proofs in \cite{ourpaper} showing that it is is NP-hard to compute the number of cops needed to catch a robber on a given digraph, even if a capture time limit is given as input.
 Since the cop number is NP-hard to compute, a possible next step would be to study fixed-parameter tractability of the problem. The first parameter that one can consider would be the cop number.  Our {next} result states that it is NP-hard to decide whether $\mbox{cn}(D) = 1$ showing that the problem cannot be FPT parameterized by the cop number. Another  consequence is that $\mbox{cn}(D)$ cannot be approximated within a factor of $2 - \epsilon$.  We show however that the problem is  FPT parameterized by $w + v_w$ where $w$ is a fixed integer and $v_w$  is the minimum size of a set whose deletion leaves an acyclic digraph where each weakly connected component contains at most $w$ vertices. We also give a tight upper bound for the number of arcs of a $k$-copwin digraph.  %Some focus is then given to computing the capture time. In particular, we prove that the capture time  is difficult to approximate within a factor of $4/3 - \epsilon$ even if the cop number is known. 
Another research direction consists in looking for special cases where the cop number is easier to compute. We first prove that this holds for the class of round digraphs for which the cop number is exactly equal to the minimum outdegree. Then we focus on tournaments and show that it is easy to check whether a tournament is $1$-copwin. We also prove that the problem becomes FPT when parameterized by the size of a minimum feedback vertex set. Finally, some connections with the binary matrix mortality problem are pointed out. This allows to show that a digraph and its reverse digraph have the same cop number and the same capture time. 

The rest of the paper is organized as follows. Some preliminary observations are presented in Section \ref{sec:prem}. Section \ref{sec:general} is dedicated to general digraphs, while tournaments are studied in Section \ref{sec:tour}. Connections with binary matrix mortality are highlighted  in Section \ref{sec:morta} and some concluding remarks follow in the last section. 

\section{Preliminaries}
\label{sec:prem}

Let $k$ be an integer and consider the digraph $\Pi_k$, whose vertices correspond to the power set of $V$, having an arc $(V_1,V_2)$ for every $V_1,V_2\subseteq V$ such that $V_2\subseteq N^+(V_1)$ and $|N^+(V_1)\setminus V_2|\le k$. {There is also an arc from $V$ to all $W\subseteq V$ with $|W|\ge |V|-k$.}

\begin{observation}\label{piconstr}
    There exists a path, in $\Pi_k$, from $V$ to $\emptyset$ if and only if $k$ cops have a winning strategy
\end{observation} 
\begin{proof}
    Every {$k$ cops }winning strategy $(W_i)_{i\ge1}$ on $D$, minimizing the capture time, defines a path $V,R_1,R_2,...\emptyset$, where $R_i$ is the robber territory at step $i$: indeed  $\Pi_k$ contains the arc $(V,R_1)$, with $R_1=V\setminus W_1$ and arcs $(R_i,R_{i+1})$ for $i\ge 1$ with $R_{i+1}=N^+(R_i) \setminus W_{i+1}$.
Similarly, for every path from $V,R_1,R_2,...,\emptyset$ a winning strategy using $k$ cops is obtained using the sets {$W_1=V\setminus R_1$ and $W_i:=N^+(R_{i-1})\setminus R_{i}$ for $i\ge 2$.}
\end{proof}
\begin{corollary}\label{ctexp}
For any digraph $D=(V,A)$, $ct(D)\le 2^{|V|} - 1$.
\end{corollary}
The digraph $\Pi_k$ can be constructed by considering every $V_1 \subseteq V$ and, if $|N^+(V_1)|\ge k$, creating an arc from $V_1$ to all $V_2$ that can be obtained as $V_2=N^+(V_1)\setminus S$, with $S\subseteq N^+(V_1)$ and $|S|\le k$. Notice that the possible choices for $S$ are certainly upper bounded by $2^{|V|}$, therefore $\Pi_k$ can be constructed in time $O(4^{|V|})$. Moreover, a path from $V$ to $\emptyset$ can be easily found in time quadratic in the number of vertices of $\Pi_k$.
\begin{corollary}\label{exposolver}
It is possible to decide whether a given digraph $D=(V,A)$ is $k$-copwin in time $O(4^{|V|})$.
\end{corollary}

Let us add another obvious lemma that is helpful to derive lower bounds for the cop number.  

\begin{lemma}
    Let $D = (V,A)$ be a digraph, $p$ and $k$ be integer numbers such that $p\le |V|$ and $\forall S \subset V$ with $|S| = p$, $|N^+(S)| \geq p + k$, then $\mbox{cn}(D) \geq k+1$.
    \label{lem:obvious}
\end{lemma}
\begin{proof}
Inequality $|N^+(S)| \geq p + k$ implies that  $p+k \le |V|$. Assume that only $k$ cops are used. The lemma hypothesis  ensures that the size of the robber territory $R_t$ will never go below $p$. Indeed, at the first time step there are at least $p$ vertices outside $W_1$, and by simple induction, starting from $R_{t-1} \supset S$ with $|S| =p$, we get that $|R_t| \ge |N^+(  R_{t-1})| - k \geq |N^+(S)| - k \ge p + k - k = p$. It is therefore  impossible to guarantee robber capture using only $k$ cops.
\end{proof}

In other words, if $\mbox{cn}(D) \leq k$, then for each $p \le |V|$, there exists $S \subset V$ such that $|S| = p$ and $|N^+(S)| \leq p + k - 1$. 
Lemma \ref{lem:obvious} can be seen as an obvious generalization of the minimum outdegree bound mentioned in Introduction obtained by taking $p=1$.

\section{General digraphs}
\label{sec:general}
{First, we correct and simplify two NP-hardness proofs provided in \cite{ourpaper}. As observed by \cite{error},  the initial proofs were implicitely assuming strong NP-hardness of the partition problem (which is obviously wrong).  We slightly modify the proofs using a reduction from $3$-partition. {An instance of $3$-partition is a multiset $S$ of $n$ positive integers $\{a_1,...,a_n\}$ with $n=3m$ and one wants to decide whether there is a partition $S_1,...,S_m$ of $S$ such that the sum of the elements in each $S_j$ equals $\beta = 1/m \sum_{i=1}^n a_i$.
This problem is NP-hard even when the $a_i$ are bounded by a polynomial in $n$ and $\frac{\beta}{4} < a_i < \frac{\beta}{2}$ for $i=1,...,n$ \cite{gareyjohnson79}. Note that the latter condition implies that the $S_j$ are triplets. {Finally, observe that we can assume $a_i > 2$ for every $i=1,...,n$, if this is not the case consider the equivalent instance $\{3a_1,...,3a_n\}$}.

\begin{theorem} 
It is NP-hard to compute the number $cn(D)$ of a given digraph $D$ even
if it does not contain loops
\label{th=cn}
\end{theorem}
\begin{proof}
Consider an instance of $3$-partition as described above. {We assume that the $a_i$ numbers are bounded above by a polynomial in $n$.}
We construct a loop-free digraph $D$ such that $cn(D)=\beta$ if and only if $\{a_1,...,a_n\}$ is a yes instance of $3$-partition.
The digraph $D$ contains, for each $i=1, ..., n$, the biorientation of the complete graph on $a_i$ vertices. Let $V_1,...,V_n$ be their vertex sets and let $V:=\bigcup_i V_i$. 
The digraph $D$ also contains $m+1$ disjoint copies of an independent set on $\beta$ vertices. Let $I_1, ...,I_{m+1}$ be their vertex sets. The digraph $D$ contains arcs from every vertex in $I_j$ to every vertex in $I_{j+1}$ for $j=1, ..., m$; arcs from every vertex in $V$ to every vertex in $I_1$ and arcs from every vertex in $I_{m+1}$ to every vertex in $V$ {(see Figure \ref{fig:3partreduc} for illustration)}.
\begin{figure}[htbp]
\centering
\vspace{-3mm}
\includegraphics
%[scale=0.27]
%{3-partreduc.pdf} 
[width=60mm,height=80mm]{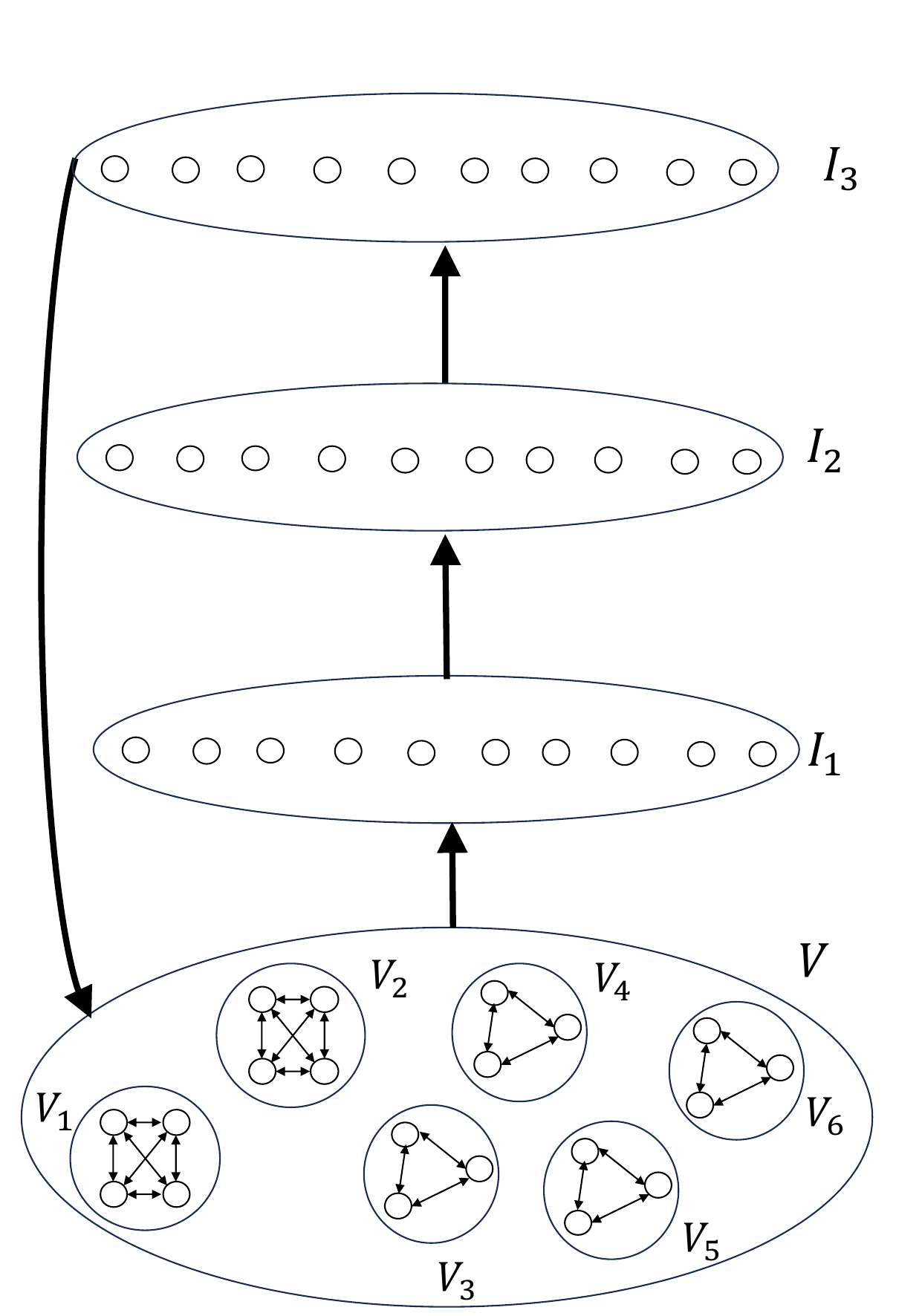}
%\vspace{-15mm}
    \caption{ Reduction from $3$-partition: $a_1 = a_2 =4$, $ a_3 = a_4 = a_5 = a_6 = 3$, $\beta = 10$, $n = 6$,$m = 2$, $|V_i|=a_i, 1 \leq i \leq 6$, $|I_j| = \beta, 1 \leq j \leq 3$.   Thick arrows indicate all arcs from a set to the other}
\label{fig:3partreduc}
    \end{figure}
We show that $cn(D) =  \beta$ if and only if the 3-partition instance is true. {Observe that the construction is polynomial in $n$ since the $a_i$ numbers are polynomially bounded. }\\
Assume that the 3-partition instance is true and consider the corresponding partitioning into triplets $S_1,...,S_m$. If $S_j = \{a_{j_1}, a_{j_2},  a_{j_3}\}$ is any  triplet of the partitioning, then $|V_{j_1}| + |V_{j_2}| + |V_{j_3}| = a_{j_1} + a_{j_2} +a_{j_3}  = \beta$. 
{A winning strategy using $\beta$ cops is as follows: at steps from $1$ to $m$ play $I_1$; at steps $m+1\le i \le 2m$ play $\bigcup_{h\in S_i}V_h$; at steps from $2m+1$ to $3m$ play $I_{m+1}$. To see that cops win, observe that $R_1=V\cup I_2 \cup ... \cup I_{m+1}; \ R_2=V \cup I_3 \cup ... \cup I_{m+1};\ ...\ ;\ R_m=V\cup I_{m+1}; \ R_{m+1}= I_1 \cup ( \bigcup_{i \not \in S_1} V_i );\ ...\ ;\ R_{2m}=I_1 \cup ... \cup I_{m};\ R_{2m+1}=I_2 \cup ... \cup I_{m};\ ... \ ; \ R_{3m}=\emptyset$. Therefore this is a winning strategy on $D$ using $\beta$ cops, which are necessary since $\delta^+(D)= \beta$.} \\
Let us now assume that $3$-partition is false and assume, by contradiction, there is a cop winning strategy  $(W_t)_{t \geq 1}$  using $\beta$ cops.  Observe that a time translation of a winning strategy  is still winning: the strategy defined by $W'_{t+c} = W_{t}$ for $t \ge 1$ and $W'_t$ is arbitrary for $1 \leq t \leq c$ is a winning strategy. Without loss of generality, we can  then consider a shortest winning  strategy under the condition  $W_t = I_1$ for the first $m+1$ time steps. Observe that $R_{m+1} = V$. 
Moreover, by optimality of the strategy, {the robber territory, from time $m+2$ on, should never intersect $I_{m+1}$} since this would allow the robber to be at any node of $V$ in the next step, contradicting the optimality of the winning strategy.
We can also assure that at time $m+2$, some cops are in $V$ and $R_{m+2} \cap I_1 \neq \emptyset$ {(since otherwise $R_{m+2}$ would contain at least the whole set $V$, contradicting again the strategy's optimality)}.  
To simplify the rest of the proof, time step   $m+2$ can be considered as our first time step. {We can then assume that $W_1 \cap V \neq \emptyset$ and $R_1 = (V \cup I_1) \setminus W_1$. }
We show by induction that  $R_{m+t} \cap V \neq \emptyset$  for any $t \geq 0$.\\
{Let us first show that at least $m+1$ time steps are needed to make $V$ robber-free.}
\begin{claim}
    {If the 3-partition instance is false and  $R_{\tau}\cap V = \emptyset$ for some time step $\tau$}, then $W_{s}\cap V\neq \emptyset$ for at least $m+1$ different time steps between $1$ and $\tau$.
    \label{cl:false}
\end{claim}
\begin{proof}
   Let $\gamma(\tau)$ be the total number of cops placed in $V$ between time step $1$ and $\tau$. Since $V_i$ is the biorientation of a complete graph, $\sum_{s=1}^{\tau} |W_s\cap V_i| \ge |V_i|$. Indeed consider the first time $t_i$ some cop is in $V_i$: if at that moment at least one vertex of $V_i$ is not occupied by a cop, then at time $t_i+1$ at least $|V_i|-1$ vertices of $V_i$ can potentially enter the robber territory and so at least other $|V_i|-1$ cops need to be on $V_i$ at a time step {between $t_i+1$ and $\tau$}. 
   Observe also that if at least two vertices of $V_i$ are not occupied by cops at step $t_i$, then the entire $V_i$ can potentially enter the robber territory at step $t_i+1$. It follows that either ${W_{t_i}=V_i}$ or $\sum_{s=1}^{\tau} |W_s\cap V_i| > |V_i|$.  \\
   We have that
   \begin{equation}\label{gamma}
       \gamma(\tau)=\sum_{i=1}^n\sum_{s=1}^{\tau} |W_s\cap V_i| \ge \sum_{i=1}^n|V_i|= \beta m.
   \end{equation} Suppose, by contradiction, that cops use vertices of $V$ in less than $m+1$ different steps. This means they do it in exactly $m$ time steps and in each time step they use exactly $\beta$ vertices from $V$. Moreover, every time some cop is in $V_i$, all the vertices of $V_i$ are occupied by cops, otherwise the inequality in (\ref{gamma}) would be strict {(recall $a_i>2$)}, which is impossible using $\beta$ cops on $V$ in at most $m$ time steps. Finally note that cops cannot occupy more than three different sets $V_i$ on the same time step because the size of each $V_i$ is strictly greater than $\frac{\beta}{4}$. We can thus conclude that cops are in $V$ for exactly $m$ different time steps $s_1,...,s_m$ and in each $s_j$ they occupy entirely some different $V_{j_1},V_{j_2},V_{j_3}$ such that $a_{j_1}+a_{j_2}+a_{j_3}=\beta$. This implies that the $3$-partition instance is true, which is a contradiction.
\end{proof}
The claim implies the result is true for $t=0$. Assume that $R_{t'} \cap V \neq \emptyset$  for any $t' \leq m+t$ and prove that $R_{m+t+1} \cap V \neq \emptyset$. 
%{Since $R_{t'} \cap V \neq \emptyset$ for any $t' \leq m+t$, if $W_{t'}\neq I_1$, then $R_{t'}$ intersects $I_1$ (even when $t'=1$ as assumed before the claim). 
%Therefore, 
For $t' \leq m+t$, in order to avoid $R_{t'+m}$ intersects %the robber territory, 
$I_{m+1}$, cops must fully occupy one of the $I_j$ no later than $m$ time steps, more precisely there exists $j(t')$ with $0 \leq j(t') \leq m$ such that $I_{1+j(t')}=W_{t'+j(t')}$. We can thus say that for any $t' \leq m+t$, there exists $j(t')$ such that $W_{t'+j(t')} \cap V = \emptyset$. Notice that  for $t''\neq t' \le m+t$, we necessarily have $t''+j(t'')\neq t'+j(t')$. {Indeed, robbers starting their  move  from $V$ to $I_{m+1}$ at different times steps cannot simultaneously be at the same $I_j$ implying that they cannot be intercepted at the same time}. 
As a consequence, if we consider 
$t' \leq t {+1}$, 
%\ale{to avoid $I_{m+1}$ becoming available (anytime before $m+t+1$)}
there will be at least $t+1$ different time steps between $1$ and $m+t+1$ where no cop is in $V$. Said another way, there are at most $m$ time steps between $1$ and $m+t+1$ where some cop is in $V$. Since this is not sufficient to cover all nodes of $V$ {(by Claim \ref{cl:false})}, $R_{m+t+1} \cap V \neq \emptyset$ proving the induction and leading to a contradiction since the strategy is not a winning one.
\end{proof}
}

{
Let $cn(D,l)$ denote the minimum number of cops required to capture the robber no later than time step $l$. We show that $cn(D,l)$ is difficult to compute. We adapt the proof of  \cite{ourpaper} by using a reduction from 3-partition instead of partition.  
} 
{
\begin{theorem}
    It is NP-hard to compute the number $cn(D,l)$ of a given digraph $D$ when $l$ is part of the input.
    \label{npcl}
\end{theorem}
\begin{proof}
Let $\{a_1,...,a_n\}$ be an instance of the $3$-partition problem where the $a_i$ are bounded by a polynomial in $n$ and are between $\beta/4$ and $\beta/2$. We still have $n = 3m$ and $\beta = (\sum_{i} a_i) 1/m$.    We construct a digraph $D$ such that $cn(D,m+1)= \beta$ if and only if $\{a_1,...,a_n\}$ is a yes instance of 3-partition.\\
The digraph $D$ contains, for each $i=1, ..., n$, the biorientation of the complete graph on $a_i$ vertices. Let $V_1,...,V_n$ be their vertex sets. The digraph $D$  contains also a biorientation of the complete graph on $\beta$ vertices. Let $Z$ denote its vertex set.  
All vertices of $D$ have also a loop.  Observe that $cn(D) \geq \delta^+(D[Z])=\beta$ implying that $cn(D,m+1) \geq  \beta$. \\
Assume that the 3-partition instance is true  and consider the corresponding partitioning into triplets $S_1, ..., S_m$. Consider the cop strategy defined by $W_i = \bigcup\limits_{j | a_j \in S_i} V_j$ for $1 \leq i \leq m$ and $W_{m+1} = Z$. This strategy allows to capture the robber no later than time $m+1$. Hence, $cn(D,m+1)=\beta$.\\
Let us now assume that the 3-partition instance is false. Assume also that the number of used cops is $\beta$. One can then use again the same arguments of Claim \ref{cl:false} proof to deduce that any winning strategy should intersect $V = \bigcup_i V_i$ during at least $m+1$ different time steps. Since we need one time step with all cops in $Z$, any winning strategy using $\beta$ cops  requires at least $m+2$ time steps. At least one more cop is needed to guarantee capture no later than time step $m+1$: $cn(D,m+1) \geq 1+\beta$. 
\end{proof}
}

By using a subdivision technique, we show it is NP-hard even to decide whether a given digraph is $1$-copwin.

\begin{theorem} \label{1ishard}
The problem of deciding wheter a given digraph $D$ is $1$-copwin is NP-hard. 
\end{theorem}
\begin{proof}
First of all, observe that we can assume input digraphs have no vertices with outdegree zero, since deleting such vertices does not change the cop number. \\
We show NP-hardness by providing a reduction from the problem of deciding, for a given pair $(D,k)$, whether $\mbox{cn}(D)\le k$. 
Given $D=(V,A)$, we subdivide all its vertices creating a path on $k$ vertices. More precisely, for every $v\in V$, the new digraph $\tilde{D}$ contains $k$ copies of $v$ denoted by $v^1,...,v^k$ and arcs $(v^i, v^{i+1})$ for $i=1,...,k-1$. The digraph $\tilde{D}$ contains also an arc $(u^k,v^1)$  for every arc $(u,v) \in A$. Notice that since $k \leq |V|$, the construction is polynomial.  We claim that $\mbox{cn}(D)\le k$ if and only if $\mbox{cn}(\tilde{D})\le 1$.\\
Assume that $\mbox{cn}(D)\le k$, let $(W_i)_{i\ge1}$ be a winning strategy using $k$ cops on $D$ that minimizes the capture time: by {Corollary} \ref{ctexp} the capture time is at most $\Omega:=2^{|V|}-1$; denote by $ v_{i_1},...,v_{i_k}$ the vertices of $W_i$, for every $1\le i \le \Omega$. \\
We will show that a single cop wins on $\tilde{D}$ by simulating $k$ times (separated by a "pausing" move) the strategy $W_1,...,W_\Omega$: the strategy is simulated using diagonal batches of subdivided vertices, where in each element of the batch the cop simulates one of the $k$ cops, more precisely the cop plays $v_{i_1}^1,...,v_{i_k}^k$, repeated for $1\le i \le \Omega$; then at step $k\Omega+1$ he plays the empty set and then replays $v_{i_1}^1,...,v_{i_k}^k$, for $1\le i \le \Omega$; at step $2k\Omega+2$ he plays the empty set and replays again, repeating until step $k^2\Omega+k-1$. 
An illustration of the construction of $\tilde{D}$ and the simulated  winning strategy is provided after the proof.
\begin{claim}
  A robber on $\tilde{D}$, starting in the first vertex of a subdivided path (a vertex $u^1$, for some $u\in V(D)$) is captured when playing against a simulation of \ $W_1,...,W_\Omega$.  
\end{claim}
\begin{proof}
 The robber must spend $k$ steps through the first subdivided vertex, then spend $k$ steps through the next subdivided vertex, etc. so his strategy is of the type $r_{1}^1,...,r_{1}^k,\\ r_{2}^1,...,r_{2}^k,...$, with $(r_i)_{i\ge 1}$ defining a robber strategy on $D$ which eventually leads to capture. This means there exist integers $t\ge 1$ and $1\le j \le k$, such that $r_t=v_{t_j}$. With respect to the game on $\tilde{D}$, at step $(t-1)k + j$ the robber is on $r_t^j$ and the cop is on $v_{t_j}^j=r_t^j$, so it gets captured.   
\end{proof}

It follows that if the robber starts in a vertex $u^1$ it gets captured in the first simulation; in case the robber first vertex is of the type $u^p$, with $2\le p \le k$, after $k-p+1$ simulations, namely at step $(k-p+1)k\Omega+k-p+1$, we will be in the hypothesis of Claim 1, since the apex of the robber vertex is congruent to $p+k-p+1=1$ modulo $k$ and the cop is just starting his $(k-p+2)$-th simulation of $W_1,...,W_\Omega$, where the robber is then captured.

Suppose now that a single cop can win on $\tilde{D}$ by playing the strategy $(\lbrace u_i^{j_i}\rbrace)_{i\ge 1}$. On $D$, $k$ cops can win by playing the vertices obtained by projecting in $D$ a consecutive batch of $k$ vertices played by the cop of $\tilde{D}$. More precisely, at step $s$, the cops play $W_s:=\cup_{j=1}^k \lbrace u_{(s-1)k+j}\rbrace$. Let $(r_i)_{i\ge 1}$ be a robber strategy in $D$ and consider a robber strategy in $\tilde{D}$, obtained by iteratively playing $r_i^1,...,r_i^k$ for $i\ge 1$: this latter strategy eventually leads to capture, meaning there exist integers $t\ge 1$ and $1\le j \le k$, such that $r_t^j=u_{(t-1)k+j}^j$, thus $r_t \in W_t$ and the robber is captured in $D$.
\end{proof}
To illustrate the technique used in the above proof,  Figure \ref{fig:subdiv} shows an example of a directed triangle with $3$ loops. To check whether the cop number is less than or equal to $2$, the graph $\tilde{D}$ is built (on the right part of the figure). An obvious winning strategy in $D$ using two cops is given by: $W_1 = \{a,b\}$ and $W_2 = \{a,c\}$. Table \ref{tab:stra} provides the simulated strategy in $\tilde{D}$ (we also give the robber territory at each time step). Observe that there is a "pausing" move at time step $5$ represented by  $\emptyset$.

\begin{figure}[htbp]
\centering
\vspace{-3mm}
\includegraphics
[scale=0.25]{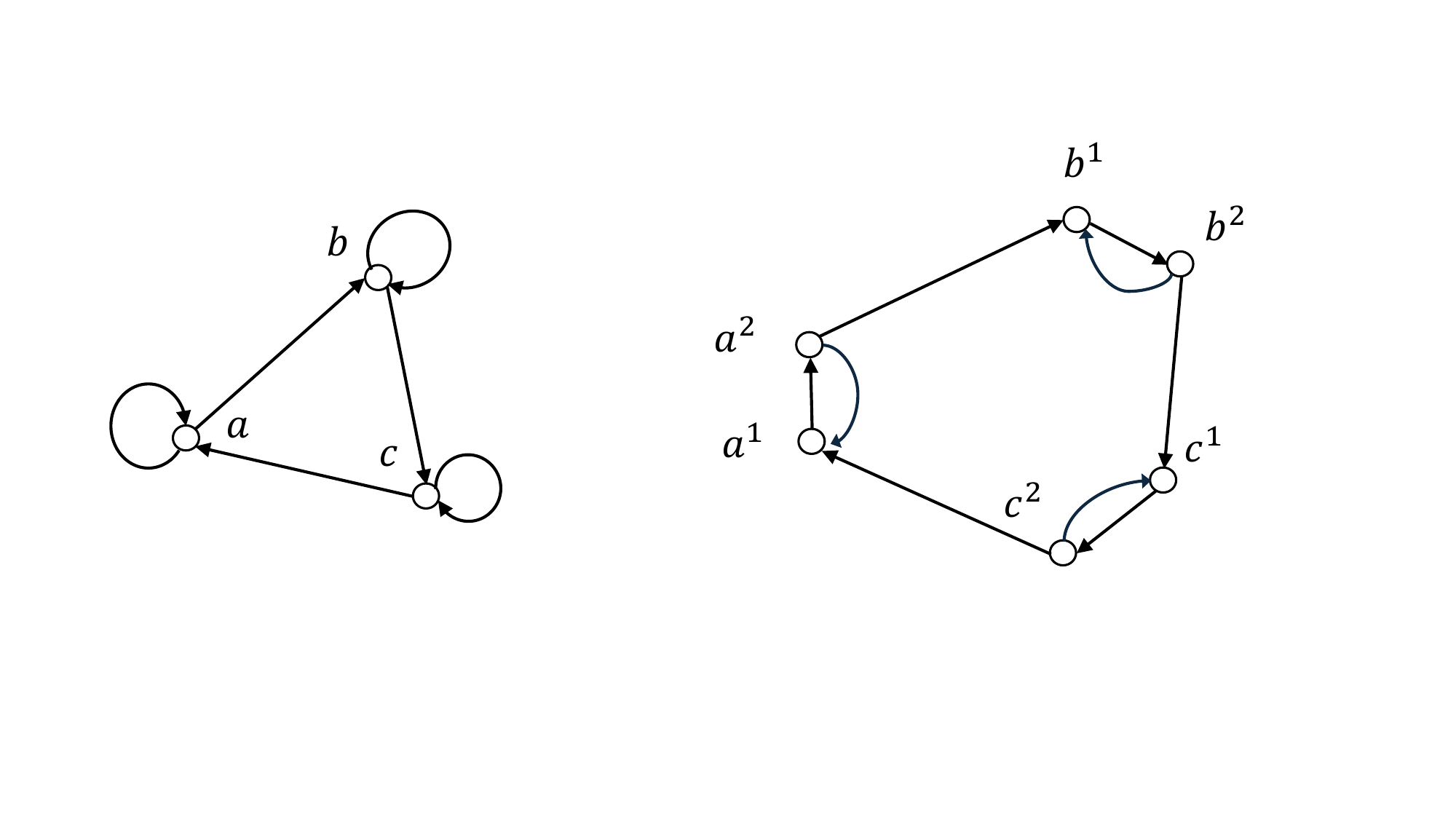} 
\vspace{-15mm}
    \caption{ Illustration of  subdivision: $k = 2$, $D$ on the left and  $\tilde{D}$ on the right}
\label{fig:subdiv}
    \end{figure}
    
\begin{table}[h!]%tbp]
\label{tab:stra}
\center
%\tiny
\hspace*{-0.3cm}
\begin{tabular}{|c|c|c|}
\hline
$t$  & $W_t$ & $R_t$ \\
\hline
$1$ & $\{ a^1 \}$ & $\{ a^2, b^1, b^2, c^1, c^2 \}$\\
$2$ & $\{ b^2 \}$ & $\{a^1, b^1, c^1, c^2 \}$ \\
$3$ & $\{ a^1 \}$& $\{a^2, b^2, c^1, c^2 \}$ \\
$4$ &  $\{ c^2 \}$& 
$\{a^1, b^1, c^1 \}$
\\
$5$ &  $\emptyset $& $\{a^2, b^2, c^2 \}$ \\
$6$ & $\{ a^1 \}$&   
$\{b^1, c^1 \}$
\\
$7$ & $\{ b^2 \}$& 
$\{c^2 \}$
\\
$8$ &  $\{ a^1 \}$& $\{c^1 \}$ \\
$9$ &  $\{ c^2 \}$& $\emptyset$\\
\hline
\end{tabular}
\caption{The simulated winning strategy on $\tilde{D}$ and the corresponding robber territory at each step.}
\end{table}

Suppose that there exists an $h$-approximation algorithm for the cop number of a digraph $D$, with $1< h < 2$: the cop number computed by this algorithm would be less than two if and only if $\mbox{cn}(D)=1$. It is therefore difficult to approximate the cop number with a factor less than two.
\begin{corollary}
For every $0 < \epsilon \le 1$, it is NP-hard to approximate $\mbox{cn}(D)$ within a factor $2-\epsilon$.
\end{corollary}
{Let us briefly remind some basic definitions of parameterized complexity. An instance of a parameterised version of a decision problem is a pair $(X, h)$, where $X$ is an instance of the decision problem and $h\ge 0$ is an integer (parameter) associated with $X$. We say that a problem is fixed-parameter tractable (FPT) if there exists an algorithm that solves a given $(X,h)$ in time that is $f(h)|X|^{O(1)}$, for any (computable) function $f(h)$. {A problem is para-NP-hard if it is NP-hard already for some constant value of the parameter.} A kernelization for a parameterized problem is an algorithm taking $(X,h)$ and producing an instance $(X',h')$ in time polynomial in\footnote{Notice this definition is equivalent to the classical definition of a kernelization algorithm that requires to run in time polynomial also in the parameter $h$} $|X|$; the produced instance $(X',h')$ is required to be a yes instance if and only if $(X,h)$ is so; the size of $X'$ and $k'$ are required to be bounded by a function of only $h$.\\
The output $(X',h')$ of a kernelization algorithm is called a kernel. It is well-known that a decidable parameterised problem is FPT if and only if it admits a kernel (see e.g. \cite{Cygan2015}).\\ \ \\}
{Another immediate corollary of Theorem \ref{1ishard} is the following.
\begin{corollary}
The problem of deciding wheter a given digraph $D$ is $k$-copwin is para-NP-hard parameterized by $k$.
\end{corollary}
}
On the positive side, we present a kernelization providing an FPT algorithm. For the undirected version of our game, the authors of \cite{dissaux_et_al:LIPIcs.MFCS.2023.42} provide an FPT algorithm parameterized by the minimum size of a vertex cover. We generalize their result.\\
Given a digraph $D=(V,A)$, let $w$ be an integer, and let $v_w$ be the minimum size of a set $X\subseteq V$, whose deletion leaves an acylic digraph having weak components of at most $w$ vertices. {A similar notion (without the requirement of the remaining graph being acyclic) has been studied for undirected graphs under the name of $w$-separator (or $w$-vertex separator)  \cite{ksep,Lee}; the problem of finding such a separator is called the component order connectivity problem and it has been studied on directed graphs from a parameterized complexity perspective in \cite{bangDCOC}. }
Notice that $v_w$ is weakly decreasing in $w$, moreover $v_1$ equals the minimum vertex cover number of $D$, while  $v_{|V|}$ equals the minimum feedback vertex set number of $D$. It follows that, for any positive integer $w$, the minimum feedback vertex set number, and thus also $\mbox{cn}(D)$, are lower bounds for $v_w$, while the vertex cover number is an upper bound. 
{
\begin{lemma}\label{approvw}
Let $D=(V,A)$ be a digraph and let $w,v_w$ be defined as above. There exists a polynomial algorithm that calculates a $(w+1)$-approximation for $v_w$ and computes the set $X$ to be deleted.
\end{lemma}
\begin{proof}
Let $G$ be the underlying undirected graph of $D$. We can begin with $X=\emptyset$ and perform a breadth first search (BFS) of the connected components of $G$ starting at any vertex: if the search terminates before reaching $w+1$ vertices, the examined component has size at most $w$, so  another vertex (from another component) of $G$ is picked and another search is performed; as soon as $w+1$ vertices are reached in the same component, they are deleted from $G$, added to $X$ and the search is continued from an untouched vertex. This procedure ends when all vertices are reached and produces a set $X$, whose deletion from $D$ leaves a digraph having weakly connected components of size at most $w$. Its running time is $O(|V|^2)$.\\
After the deletion of $X$, in order to produce an acyclic digraph, a (directed) BFS is performed on each weak component of the remaining directed graph (where all cycles have length at most $w$): starting from any vertex, its component is explored and, as soon as a directed cycle is found, all the vertices of the cycle are deleted and added to $X$; the search is then restarted on the same (updated) component and continues until the component becomes acyclic. This is repeated on all the weak components. Notice that, since there are $O(|V|)$ components (of size at most $w$) and since at most $w$ BFS are performed on the same component, the running time of this second procedure is $O(|V|w^3)$. 
The total running time of the two procedures is therefore $O(|V|^2 + |V|w^3)$. \\
The set $X$ calculated at the end is such that the digraph obtained after the deletion of $X$ is acyclic and has weak components of size at most $w$. Now consider the disjoint groups of vertices deleted by any of the two procedures: in the first one they are a weakly connected subgraph on $w+1$ vertices, in the second one they form a directed cycle of size at most $w$. It follows that any set whose deletion leaves an acyclic digraph with weak components of size at most $w$, must contain at least a vertex from each of the deleted group vertices. We thus have $|X| \le (w+1)v_w$.
\end{proof}}
\begin{theorem}
Let $D=(V,A)$ be a digraph, let $k$ be an integer and let $w,v_w$ be defined as above. Deciding whether $\mbox{cn}(D)\le k$ is FPT parameterized by $w+v_w$.
\end{theorem}
\begin{proof}
{Let $X$ be a set as calculated in Lemma \ref{approvw}: } consider the digraph obtained by the deletion of $X$
and let $C_1,...,C_l$ be its weak components. We define an equivalence relation between components as follows: two components are equivalent if they are isomorphic and the isomorphism matches vertices having the same in- and out-neighbors in $X$. Notice that the number of possible different classes of equivalence is upper bounded by
%$$f(w,v_w):=\sum_{i=1}^w (2^{ (w+1)v_w}\cdot 2^{(w+1)v_w})^i \cdot 2^{i^2}\le \sum_{i=1}^w (2^{ 2(w+1)v_w+1})^{i^2}< 2^{(2(w+1)v_w+1)(w^2+1)}.$$
$$f(w,v_w):=\sum_{i=1}^w (2^{ (w+1)v_w}\cdot 2^{(w+1)v_w})^i \cdot 2^{i^2}\le 2^{w^2} \cdot\sum_{i=1}^w (2^{ 2(w+1)v_w})^{i} < 2^{w^2+1+ 2(w+1)v_w w}.$$  
Consider a digraph $D'$ obtained from $D$ by keeping only $kw+1$ components from the classes that contain more than $kw+1$ components. 
\begin{claim}
    We have that {$\mbox{cn}(D)\le k$ if and only if $\mbox{cn}(D')\le k$}
\end{claim}
\begin{proof}
{ The first implication follows from the fact that $D'$ is a subdigraph of $D$, hence $\mbox{cn}(D')\le \mbox{cn}(D)$. For the converse, suppose by contradiction that $\mbox{cn}(D') \le k < \mbox{cn}(D)$} and consider a cop winning strategy $(W_i)_{i\ge1}$ in $D'$ using $\mbox{cn}(D')$ cops: the cop strategy is also a strategy on $D$, where there must exist a robber strategy $(r_i)_{i\ge1}$ that can escape it. Now, if some $r_t$ is not a vertex of $D'$, then there exists $0\le h \le w-1$ such that $r_t,...,r_{t+h}$ are the vertices of a path in a component $C$ that has been discarded from $D$ and $r_{t+h+1}\in X$ ($r_{t-1} \in X$ too, if $t\ge 1$). The digraph $D'$ contains $wk+1$ components equivalent to $C$, thus there must be at least one that is cop-free at steps $t,...,t+h$. Let $C'$ be such a component and let $r'_t,...,r'_{t+h}$ be the vertices of $C'$ matched to $r_t,...,r_{t+h}$ respectively, by the isomorphism between $C$ and $C'$. Note that, by definition of our equivalence, $r_{t+h+1}$ is an out-neighbor of $r'_{t+h}$ and, if $t\ge1$, $r_{t-1}$ is an in-neighbor of $r'_{t}$. It follows that by replacing $r_t,...,r_{t+h}$ with $r'_t,...,r'_{t+h}$ we still have an escaping robber strategy. All robber vertices outside $D'$ can be replaced in this way by vertices of $D'$ obtaining a robber strategy in $D'$ that can escape the cop strategy $(W_i)_{i\ge1}$, contradicting the assumption that there was no such robber strategy.
\end{proof}
Our kernelization algorithm, for a fixed $w$, is as follows: given an input $(D,k)$, first calculate a $(w+1)$-approximation for $v_w$ and the corresponding set of vertices $X$ to be deleted; if $|X|\le k$, then output a YES-instance. Otherwise construct the digraph $D'$ and output $(D',k)$. Observe that the construction of $D'$ can be done by {checking, for each possible pair of components (of size at most $w$), whether there exists an isomorphism between them that matches vertices having the same in- and out-neighbors in $X$. Now, there are $O(|V|^2)$ possible pairs of components and, for each pair, $O(w!)$ possible isomorphisms; moreover checking neighbors in $X$ can be done in time $O(w(w+1)v_w)$ {given that the size of a component is at most $w$ and the size of $X$ is at most $(w+1)v_w$}. It follows that $D'$ can be constructed in time $O\left(w!w^2v_w \cdot |V|^2 \right)$.
Recall that $X$ can initially be calculated (as in Lemma \ref{approvw}) in time $O(|V|^2 + |V|w^3)$, so the kernelization algorithm total running time is $O\left(w!w^2v_w \cdot |V|^2 \right)$.\\
The number of vertices of $D'$ is at most $$(w+1)v_w + (wk+1)f(w,v_w)\le  (w+1)v_w + (w(w+1)v_w+1)f(w,v_w)=:g(w,v_w),$$
where the inequality is obtained using that $k\le |X| \le (w+1) v_w$. Observe that we can use Corollary \ref{exposolver} to solve the problem on $D'$ in time $2^{O\left( g(w,v_w)\right)}$ and $g(w,v_w)$ is $O\left( w^2v_wf(w,v_w)\right)$, namely $2^{O\left( w^2v_w\right)})$. It follows that the total running time of our FPT algorithm is 
$O\left(w!w^2v_w \cdot |V|^2 + 2^{2^{O\left( w^2v_w\right)}}\right)$.} 

%if their vertices have the same in- and out-neighbors in $X$. This can be done \ale{by testing for isomorphism $O(|V|^2)$ pair of acyclic digraphs of size at most $w$. Acyclic digraphs isomorphism can be polynomially reduced to undirected graph isomorphism \cite{zemlaiso}, that can be solved in time $2^{O\left( \sqrt{w\log w }\right)}$ \cite{babai83}. Recall that the two initial procedures have a running time $O(|V^2|)$, so the kernelization algorithm total running time is $2^{O\left( \sqrt{w\log w }\right)}\cdot |V|^2$.
%The number of vertices of $D'$ is at most $$(w+1)v_w + (wk+1)f(w,v_w)\le  (w+1)v_w + (w(w+1)v_w+1)f(w,v_w)=:g(v,v_w),$$
%where the inequality is obtained using that $k\le |X| \le (w+1) v_w$. Observe that we can use Corollary \ref{exposolver} to solve the problem on $D'$ in time $2^{O\left( g(v,v_w)\right)}$ and $g(v,v_w)$ is $O(w^2v_w2^{O\left( w^2v_w\right)})$. It follows that the total running time of our FPT algorithm is $2^{O\left( \sqrt{w\log w }\right)}\cdot |V|^2  + 2^{2^{O\left( w^2v_w\right)}}$.}
\end{proof}

Since it is NP-hard to compute the cop number, we provide here an easy-to-check necessary condition by computing the maximum number of arcs of a $k$-copwin digraph.
\begin{theorem}
If $\mbox{cn}(D) \leq k$, then $| A(D)| \le  nk + (n-k)(n+k-1)/2$,
%$k^2 + 2k(n-k) + (n-k)(n-k-1)/2$, 
and the bound is reached for some digraphs.  
\end{theorem}
\begin{proof}
We prove the result by induction on $n:=|V(D)|$. If $D$ is  k-copwin then $n \ge k$. To initialize the induction, observe that when $n = k$, the upper bound is just $n^2$ which is obviously valid. 
Let $D = (V,A)$ be a k-copwin digraph and assume that $|A(D)| \geq 1 + nk + (n-k)(n+k-1)/2$. Then $\sum_{v \in V} d^+(v) + d^-(v) \ge 2 + 2nk + (n-k)(n+k-1)$, where $d^+(v)$ (resp. $d^-(v)$) denotes the outdegree (resp. indegree) of $v$. 
 We know from Lemma  \ref{lem:obvious} that  $\mbox{cn}(D) \le k$ implies that there exists at least one vertex $v_0$ such that $d^+(v_0) \le k$.
 We will consider the digraph $D' = D \setminus \{v_0\}$. Observe that $D'$ is also $k$-copwin and consider the sum $\sum_{v \in V \setminus \{v_0 \}} d^+_{D'}(v) + d^-_{D'}(v)$ (we added here the subscript $D'$ to emphasize the fact that degrees are here considered in $D'$ and not in $D$). 
Assume first that $d^-(v_0) = n = |V|$. This implies that $d^+(v_0) + d^-(v_0) \leq n + k$ and $D$ contains a loop around $v_0$. By deleting the vertex $v_0$, the sum of outdegrees will decrease by at most $(n+k)$ (upper bound of $d^+(v_0) + d^-(v_0)$) plus $(n-1+k-1)$ (the outdegree of each vertex decreases by $1$ while the indegree of at most $k-1$ vertices will decrease by $1$). In other words, 
$\sum_{v \in V \setminus \{v_0 \}} d^+_{D'}(v) + d^-_{D'}(v) \geq (\sum_{v \in V} d^+(v) + d^-(v))  - (n+k) - (n-1+k-1) \ge 2 + 2nk + (n-k)(n + k -1) - 2(n+k) + 2 = 2 + 2(n-1)k + (n-1 - k)(n-1+k-1)$. This leads to a contradiction with the induction hypothesis.
Assume now that $d^-(v_0) \le n-1$. Then the sum of outdegrees will be reduced by at most $(n-1 + k)$ (since $v_0$ is deleted and $d^+(v_0) + d^-(v_0) \le  n-1 + k$)  and another $(n-1+k)$ (at most $k$ vertices of $V\setminus \{v_0\}$ have their indegree reduced by $1$ and at most $n-1$ vertices    have their outdegree reduced by $1$). We get again $\sum_{v \in V \setminus \{v_0 \}} d^+_{D'}(v) + d^-_{D'}(v) \ge 2 + 2nk + (n-k)(n + k -1) - 2(n+k) + 2 = 2 + 2(n-1)k + (n-1 - k)(n-1+k-1)$ leading to the same contradiction. \\
To prove that the upper bound is tight, let us consider the digraph $D=(V,A)$ defined as follows:  $V = S \cup T$ with $|S|= k$, $|T|=n-k$, the digraph $D[S]$ induced by $S$ is complete ($k^2$ arcs), $D[T]$ is an acyclic tournament (it contains $(n-k)(n-k-1)/2$ arcs and no cycles) and $D$ contains all possible arcs from $S$ to $T$ and from $T$ to $S$. By putting $k$ cops on $S$ vertices, the robber is captured so $D$ is $k$-copwin. Observe that the number of arcs is $k^2 + (n-k)(n-k-1)/2 + 2k(n-k) = nk + (n-k)(n+k-1)/2$. 

\end{proof}

{
A class of digraphs for which computing the cop number is easy, is that of round digraphs. A digraph on $n$ vertices is {round} if we can label its vertices $v_1,...,v_n$ so that for each index $i$, we have $N^+(v_i)= \lbrace v_{i+1},...,v_{i+d^+(i)} \rbrace$ and $N^-(v_i)= \lbrace v_{i-d^-(i)},...,v_{i-1} \rbrace$ (indices taken modulo $n$). We call the labeling $v_1,...,v_n$ a {round ordering}. Round digraphs can be recognized (and a round ordering 
provided) in polynomial time \cite{huangAJC19}.
A simple polynomial characterization, in terms of minimum outdegree, exists for round digraphs cop number. Let $\mbox{fvs}(D)$ be the size of a minimum feedback vertex set of $D$.
\begin{theorem}
For a round digraph $D$, $\delta^+(D)=\mbox{cn}(D) = \mbox{fvs}(D)$
\end{theorem}
\begin{proof}
We have already seen that, for any digraph $D$,
$\delta^+(D)\le \mbox{cn}(D)\le \mbox{fvs}(D)$. We will see that $\delta^+(D)=\mbox{fvs}(D)$, if $D$ is round. \\
First let us assume $\delta^+(D)<n$, otherwise the result follows immediately.
Let $v_1,...,v_n$ be a {round ordering} and assume, without loss of generality, that $d^+(v_n)=\delta^+(D)$, namely $N^+(v_n)= \lbrace v_{1},...,v_{\delta^+(D)} \rbrace$ . Observe that there can be no arc between $v_i$ and $v_j$ if $\delta^+(D)<j<i\le n$, since, in this case, $N^-(v_j) \supseteq \lbrace v_i,...,v_n,v_1,...,v_{j-1}\rbrace $, thus $v_j\in N^+(v_n)$, which is not possible for $\delta^+(D)<j <n$. It follows that the vertices $v_{1+\delta^+(D)}, ..., v_n$ form an acyclic ordering of $V(D) \setminus N^+(v_n)$ and so $N^+(v_n)$ is a feedback vertex set of $D$, of size $\delta^+(D)$.  
\end {proof}
}

%1.  The cop number k   is NP-hard to compute, and it is NP-complte to decide if k==1

%2. The cop number is non approximable within 2-eps

%3. The problem is FPT in  the parameter of vertex cover

%4. cop number of digraph is equal to the cop number of its reverse  digraph

%5.  cop number is greater then the max min out degree (and indegree) and smaller than the size of the feedback vertex set f

%6. For some graphs (I don’t remember their names, but I think vertex vi was connotes to vj if j-i is less than k modulo n), the cop number is easy to compute….

%7. The maximum number of edges of a k-copwin graph  can be computed (graphs reaching the bound are given). 

%8. If we focus on tournaments then we have:

%8.1 a  cop-nmber of a strongly-connected tournament is 1  if and only if the feedback vertex set  is equal to 1 (I checked and simplified a little-bit the proof). A corollary is that f=2 implies k = 2. 

%8.2 if the cop number is k>= 2, then f can be any value greater or equal than f  (for example, by the constuction  T1, T2,…,Tn-1, Tn where we orient from Ti to Tj whenever j>I,  except Tn from T1, and all Ti are triangles, except Tn that is only one vertex, we get a graph with k = 2, anf  f =n..)

%8.3WXGFFCCFBCF                            The lower bound can also be far from the cop number since we can build strongly connected tournaments  where  the lower bound is 1 while the cop number is n-1/2
\section{Tournaments}
\label{sec:tour}

%\begin{figure}[htbp]
%\centering
%\vspace{-10mm}
%\includegraphics[width=121mm,height=70mm]{claim1_cases.pdf}
%[scale=0.45]{claim1_cases.pdf} 
%\vspace{-10mm}
 %   \caption{The four cases related to the first claim depending on the orientation of $\{v_4,v_1\}$ and 
%    $\{v_4, v_2\}$}
%    \label{fig:firstclaim}
%    \end{figure}

Let us now focus on tournaments. Remember that a tournament is a digraph with exactly one arc between each two vertices, in one of the two possible directions.  Since the cop number of a digraph is the maximum of the cop numbers of its connected components, we will focus on connected tournaments. 
    
\begin{theorem}
A connected tournament $T$ has a feedback vertex set of size one if and only if $\mbox{cn}(T)=1$.
\label{th:cn=1}
\end{theorem}
\begin{proof}
Since $T$ is connected, it contains cycles implying that $\mbox{cn}(T) \geq 1$. 
The existence of a feedback vertex set of size $1$ obviously leads to $\mbox{cn}(T) = 1$. Let us then prove the other direction and assume that $\mbox{cn}(T) = 1$. Lemma \ref{lem:obvious} implies the existence of a subset of vertices $S$ of cardinality $n-1$ such that $|N^+(S)| \le n - 1$.
Let $w$ be the unique vertex of  $V \setminus S$. By connectivity of $T$, $N^+(w) \neq \emptyset$ and $N^-(w) \neq \emptyset$. Since $|N^+(S)|\le n-1$ and $w \in N^+(S)$, there exists $v \in S$ such that $v \notin  N^+(S)$. In other words, $S \subset N^+(v)$. Again by connectivity, $T$ should contain the arc $(w,v)$. 
If $T[S]$ is acyclic then all cycles of $T$ go through $w$ implying that 
$\{w\}$ is a feedback vertex set establishing the wanted result.  Let us then assume that $T[S]$ contains cycles. \\ Since $T[S]$ is a tournament containing cycles, $T[S]$ contains triangles.  Let $T[\{v_1, v_2, v_3\}]$ be any  triangle of $T[S]$ (so $T$ contains the arcs $(v_1, v_2)$, $(v_2, v_3)$ and $(v_3, v_1)$)  
and consider the distance to $w$ denoted by $d(T[\{v_1, v_2, v_3\}],w)$ and defined by $d( T[\{v_1, v_2, v_3\}],w) = \min(d(v_1,w), d(v_2,w), d(v_3,w))$, where $d(v_i,w)$ is  the length of a shortest path from $v_i$ to $w$.

Let us  first show that $d(T[\{v_1, v_2, v_3\}],w)$ cannot be equal to $2$. 
\begin{claim}
$d(T[\{v_1, v_2, v_3\}],w) \neq 2$.
\end{claim}
\begin{proof}
Assume for contradiction that $d(T[\{v_1, v_2, v_3\}],w)  =2$. Without loss of generality, we can assume that $v_3$ is the closest vertex to $w$: $d(T[\{v_1, v_2, v_3\}],w) = d(v_3, w)$. Let $v_4$ be a vertex such that $v_3 v_4 w$ is a shortest path. We will focus on the subgraph $T'= T[\{v_1, v_2, v_3, v_4, w, v\}]$.
Remember that $T$ contains the arc $(w,v)$ and all arcs from $v$ to the rest of vertices {and observe that $v\not \in \{v_1,v_2,v_3,v_4\}$}. Since $d(T[\{v_1, v_2, v_3\}],w)  =2$, $T$ contains the arcs $(w ,v_1)$, $(w, v_2)$ and $(w, v_3)$. {Consider any subset of vertices $U \subset \{v_1, v_2, v_3, v_4, w, v \}$ of size $3$. We will prove  that $|N^+_{T'}(U)|\ge 4$ which allows to conclude that $\mbox{cn}(T') >1$ (from Lemma \ref{lem:obvious}) and consequently $\mbox{cn}(T) >1$ contradicting the assumption that $T$ is $1$-copwin. If $U \ni w$, then  $N^+_{T'}(U) \supset \{v_1, v_2, v_3, v\}$. If $U \ni v$, then $N^+_{T'}(U) \supset \{v_1, v_2, v_3, v_4\}$. If $U = \{v_1, v_2,v_3 \}$, then $N^+_{T'}(U) = \{v_1, v_2, v_3, v_4\}$. If $U = \{v_1, v_2,v_4 \}$, then $N^+_{T'}(U)$ contains $v_2,v_3,w$ and one among $v_1,v_4$, depending on the direction of the arc between $v_1$ and $v_4$. If $U = \{v_1, v_3,v_4 \}$, then $N^+_{T'}(U)$ contains $v_1,v_2,v_4,w$. If $U = \{v_2, v_3,v_4 \}$, then $N^+_{T'}(U)$ contains $v_1,v_3,w$ and one among $v_2,v_4$, depending on the direction of the arc between $v_2$ and $v_4$. } 
\end{proof}
Having  proven the claim, we can even certify that for any triangle $T[\{v_1, v_2, v_3\}]$ of $S$, we have $d(T[\{v_1, v_2, v_3\}],w)  =1$.
\begin{claim}
$d(T[\{v_1, v_2, v_3\}],w)  =1$ for any triangle $T[\{v_1, v_2, v_3\}]$ of $T[S]$.
\end{claim}
\begin{proof}
Using the previous claim, we only have to show that we cannot have a triangle such that $d(T[\{v_1, v_2, v_3\}],w)  \ge 3$. Assume for contradiction that such a triangle exists and let $v_3 v_4... v_l$ ($v_l = w$ and $l \geq 6$)  be a shortest path whose length is equal to $d(T[\{v_1, v_2, v_3\}],w) = l-3$. Observe that $T$ does not contain the arcs $(v_2, v_4)$, $(v_3, v_5)$, $(v_4, v_6)$,...(since otherwise $v_3, v_4, ..., v_l$  would not be a shortest path). Then $T$ contains $(v_4, v_2)$, $(v_5, v_3)$, $(v_6, v_4)$, ..., $(v_l, v_{l-2})$. This leads to the existence of other triangles of $S$ that are closer to $w$ ($T[\{v_2, v_3, v_4\}]$, $T[\{v_3, v_4, v_5\}]$, $T[\{v_4, v_5, v_6\}]$, etc.) 
Note that the triangle $T[\{v_{l-4}, v_{l-3}, v_{l-2}\}]$ is at distance $2$ to $w$ contradicting the previous claim.
\end{proof}

We know from the previous claim that 
$d(T[\{v_1, v_2, v_3\}],w)  =1$. The next claim shows that there is only one vertex in the triangle that is at distance $1$ to $w$.

%\begin{figure}[htbp]
%\centering
%\vspace{-10mm}
%\includegraphics
%[width=121mm,height=50mm]{claim3_cases.pdf}
%[scale=0.30]{claim3_cases.pdf} 
%\vspace{-10mm}
 %   \caption{The three cases related to the third claim depending on the orientation of the edges $\{w, v_1\}$ and $\{w,v_2\}$}
    %\label{fig:thirdclaim}
    %\end{figure}

\begin{claim}
    If $d(T[\{v_1, v_2, v_3\}],w)  = d(v_3,w)$, then $T$ contains the arcs $(w, v_1)$, $(w, v_2)$ and $(v_3, w)$.
\end{claim}
\begin{proof}
{First observe that $v\neq v_1,v_2,v_3$ since it cannot belong to any triangle of $T[S]$.} Consider the subgraph $T'=T[\{v_1, v_2, v_3, w,v  \}]$. From the previous claim, we have $d(T[\{v_1, v_2, v_3\}],w)  = 1$ implying that $T$ contains the arc $(v_3, w)$. 
{Assume that $T$  contains 
$(v_1, w)$ or $(v_2, w)$. 
Let us prove that for each subset $U \subset \{v_1, v_2, v_3, w,v\}$ of size $2$, $|N^+_{T'}(U)| \ge 3$ implying that $\mbox{cn}(T') >1$.
If $v \in U$, then $N^+_{T'}(U)\supset N^+_{T'}(\{v\}) = \{ v_1, v_2, v_3\}$. 
Observe also that if $U=\{v_i,v_j\}$, with $1\le i\neq j \le 3$, then  $N^+_{T'}(U)$ contains $v_{i+1},v_{j+1}$ (indices modulo $3$) and $ w$. Finally, if $U=\{v_i,w\} $, then $\{v_{i+1},v\} \subset N^+_{T'}(U)$ and either $v_i$ or $w$ belong to $N^+_{T'}(U)$.
}
\end{proof}

\begin{figure}[htbp]
\centering
\includegraphics
%[width=121mm,height=50mm]{claim3_cases.pdf}
[scale=0.20]{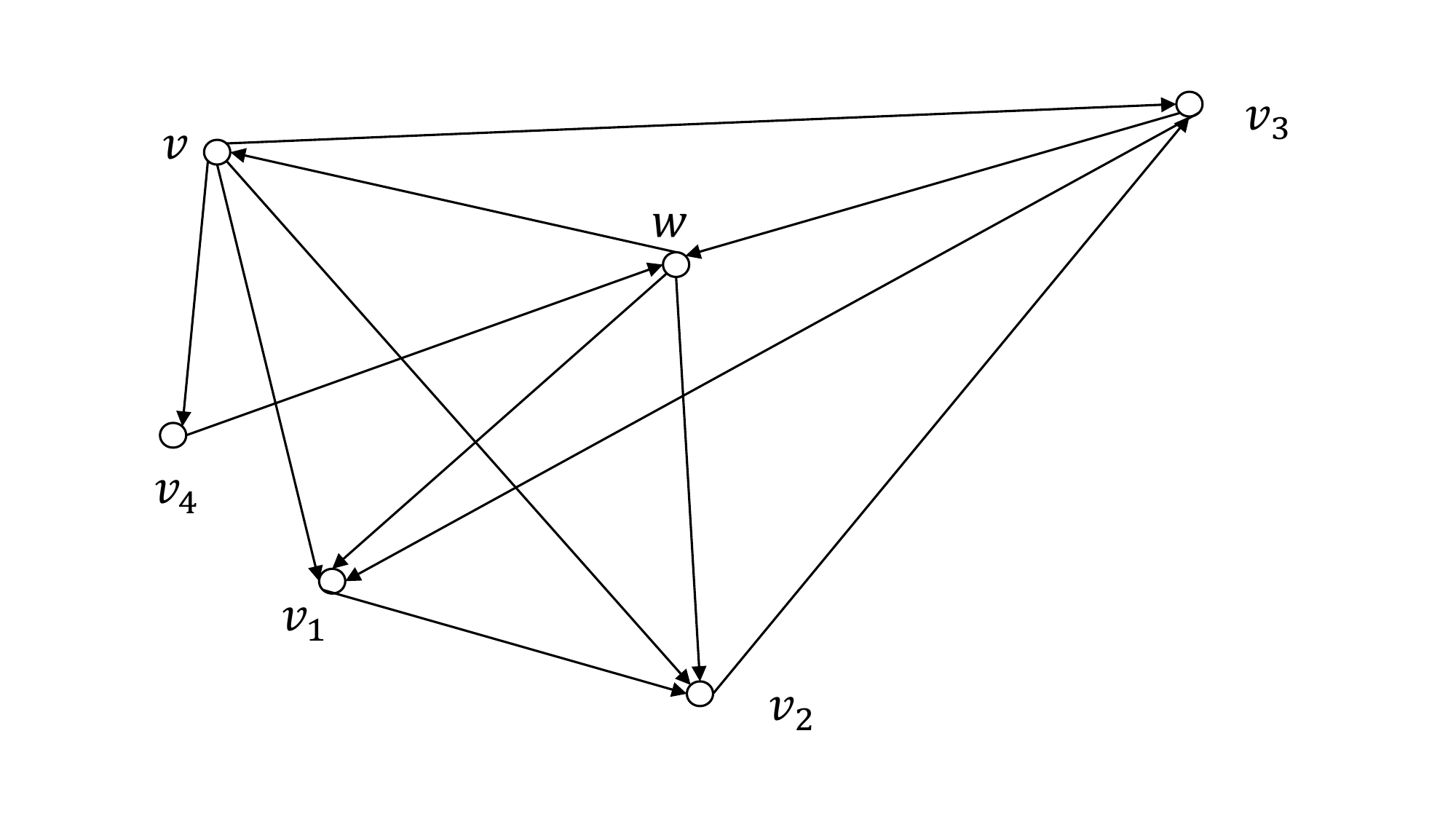} 
\vspace{-5mm}
    \caption{The digraph related to the proof of Theorem \ref{th:cn=1}}
    \label{fig:lastcase}
    \end{figure}
Observe that if $w$ has only one in-neighbor (say $v_3$), then all triangles of $T[S]$ (and thus all cycles) contain $v_3$. A cycle going through $w$ should also contain $v_3$.  Hence $\{v_3\}$ is a feedback vertex set of $T$.   Assume now that $w$ has at least another in-neighbor (say $v_4$)  in addition to $v_3$ and we still have $d(T[\{v_1, v_2, v_3\}],w)  =1$. From the previous claim, we know that $T$ contains $(w, v_1)$,  $(w, v_2)$  and $(v_3, w)$ (see Figure \ref{fig:lastcase}). 
We will focus again on the subgraph $T'= T[\{v_1, v_2, v_3, v_4, w, v\}]$ and show that  for every subset of vertices $U \subset \{v_1, v_2, v_3, v_4, w, v \}$ of size $3$, $|N^+_{T'}(U)|\ge 4$ occurs.  If $U$ contains  $v$, the property obviously holds. If $U = \{v_1,v_2,v_3\}$, then $N^+_{T'}(U) \supset \{v_1, v_2, v_3, w  \}$. If $U = \{v_1, v_2, v_4  \}$, then $N^+_{T'}(U) \supset \{ v_2, v_3, w  \}$ in addition to either $v_1$ or $v_4$ (depending on the direction of the arc between $v_1$ and $v_4$).
If $U = \{v_2, v_3, v_4  \}$, then $N^+_{T'}(U) \supset \{ v_3, v_1, w  \}$ in addition to either $v_2$ or $v_4$ (depending on the direction of the arc between $v_2$ and $v_4$).
If $U = \{v_1, v_3, v_4  \}$, then $N^+_{T'}(U) \supset \{ v_2, v_1, w  \}$ in addition to either $v_3$ or $v_4$ (depending on the direction of the arc between $v_3$ and $v_4$). Let us now consider the case where $U$ contains $w$. Observe that $N^+_{T'}(U) \supset N^+_{T'}(w) = \{v_1, v_2, v\}$. If we add to $U$ either $v_3$ or $v_4$,  then $N^+_{T'}(U)$ will also contain $w$ leading to $|N^+_{T'}(U)|\ge 4$. If $v_2$ is included in $U$, then $v_3$ is added to $N^+_{T'}(U)$ inducing the same conclusion. In other words, any subset $U \subset \{v_1, v_2, v_3, v_4, w, v \}$ of size $3$ containing $w$ satisfies $|N^+_{T'}(U)|\ge 4$.   
This implies that $\mbox{cn}(T) \geq 2$  ending the proof.
\end{proof}

\begin{corollary}
Given a tournament $T$, if a minimum feedback vertex set of $T$ has size two, then $\mbox{cn}(T)=2$.
\end{corollary}
Notice that the converse of this corollary is not true, indeed, for every $n\ge 3$ there exists a tournament $T_n$ on $3n-2$ vertices with $\mbox{cn}(T_n)=2$ and such that any feedback vertex set of $T_n$ has size at least $n$. The vertices of $T_n$ can be partitioned into $n-1$ directed triangles $C_1,...,C_{n-1}$ and a single vertex $v$. Besides the arcs of the triangles, $T_n$ contains also arcs from all the vertices of $C_i$ to all the vertices of $C_j$ for $1\le i < j \le n-1$; arcs from $v$ to all vertices of $C_i$, for $i=1,...,n-2$ and arcs from all vertices of $C_{n-1}$ to $v$. {See Figure \ref{lowcnhighfvs} for an example with $n=4$}. It is easy to see that a feedback vertex set must contain $v$ and at least a vertex from all of the $C_i$. On the other hand, two cops can win in $3n-3$ steps by playing one cop always in $v$ and another cop picking the same vertex on $C_i$ for three time steps, for $i=1,...,n-1$.
\begin{figure}[htbp]
\centering
\vspace{-3mm}
\includegraphics
[scale=0.25]{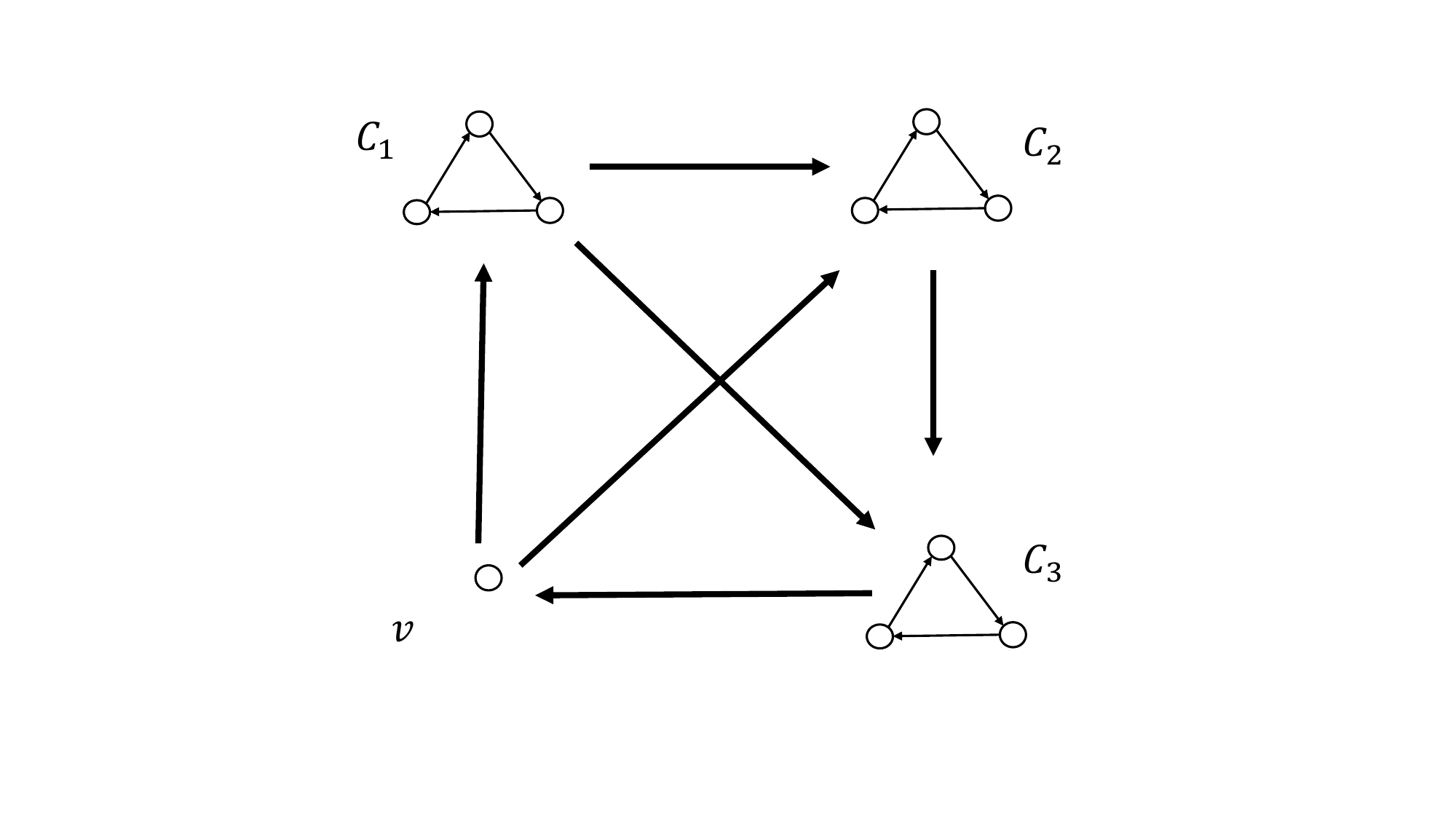} 
\vspace{-5mm}
    \caption{A $2$-copwin tournament of order $3n-2$ (here $n=4$) whose minimum feedback vertex set is of size $n$ (a thick arrow from $C_i$ to $C_j$ represents all arcs from $C_i$ vertices to $C_j$ vertices)}
    \label{fig:k=2}
    \end{figure}\label{lowcnhighfvs}

The existence of a polynomial algorithm to determine whether a given tournament $T=(V,A)$ has $\mbox{cn}(T)=1$ also follows from Theorem \ref{th:cn=1}: {Tarjan's algorithm  \cite{doi:10.1137/0201010} can be used to compute the strong components of $T$ in time $O(|V|+|A|)$ and for each component it is possible }\cite{lokshtanov2016linear} to decide whether there exist a feedback vertex set of size one in time $O(|V|+|A|)$. So we have the following corollary.
\begin{corollary}
Given a tournament $T=(V,A)$, it is possible to determine whether $\mbox{cn}(T)=1$ in time $O(|V|+|A|)$.
\end{corollary}

When the size $f$ of a minimum feedback vertex set of $T$ is bounded by a constant, one can decide whether $\mbox{cn}(T)=k$ in polynomial time: we provide an FPT algorithm parameterized by the size of a minimum feedback vertex set.\\
%The idea is to construct a transition digraph, where vertices are subsets of $V$ and represent robber territories and arcs are possible transition from one territory to another: here there is a path from the vertex corresponding to $V$ to the empty set vertex if and only if $k$ cops have a winning strategy. By using a transition digraph with only a limited number of representative subsets, one can obtain the desired FPT algorithm.
 \begin{theorem}
 Given a tournament $T=(V,A)$ and an integer $k$, determining whether $\mbox{cn}(T)=k$ is FPT parameterized by the size of a minimum feedback vertex set of  $T$.
 \end{theorem}
 \begin{proof}
{Consider the construction of the transition digraph $\Pi_k$ (as defined in Section \ref{sec:prem}) applied to the tournament $T$ and let $V_1,V_2\subseteq V$ be two of its vertices}. We say that $V_1$ is equivalent to $V_2$ if
%\footnote{in all the proof $N^+ \ (N^-)$ refers to the out- (in-) neighborhood in $T$.} 
 $N^+(V_1)=N^+(V_2)$, where $N^+ \ (N^-)$ refers to the out- (in-) neighborhood in $T$. 
 Let $F\subseteq V$ be a feedback vertex set of $T$ of size $f$, the number of equivalence classes is upper bounded by $4^f\cdot (n-f+1)$.\\Indeed let $v_1,...,v_{n-f}$ (where $n=|V|$) be the topological ordering of $V\setminus F$. The equivalence class of any $R\subseteq V$ is completely determined by $N^+(R)\cap F, \ N^+(R\cap F)$ and the vertex (if any), $v_{i_0}$, which ranks first in the topological order of  $R\setminus F$. There are at most $2^f$ possible choices for $N^+(R) \cap F$, at most $2^f$ possible choices for $R \cap F$, and $n-f+1$ choices for $v_{i_0}$ (including the choice $R\setminus F=\emptyset$). \\
 Observe that two equivalent subsets have the same out-neighbors in $\Pi_k$, therefore there is a path in $\Pi_k$, from $V$ to $\emptyset$ if and only if there is a path from the class of $V$ to that of $\emptyset$ in the quotient digraph. Therefore, using Observation \ref{piconstr}, it suffices to show that one can construct this quotient digraph in polynomial time to have our FPT algorithm.\\
 To construct the vertices of the quotient, we can build an exhaustive list of candidates by considering all possible triples $(F',F'',v_{i_0})$, where $F',F''\subseteq F$ and $v_{i_0} \in V\setminus F$ or is blank. In order to check whether a triple can be associated to a representative $R$, with $F'=N^+(R)\cap F$, $F''=R\cap F$ and $v_{i_0}$ (if not blank) ranking first in the topological ordering of $R\setminus F$, one can start computing the set $N^-(F')\cup v_{i_0}$: delete from it all vertices (if any) having an out-neighbor in $F\setminus F'$; then delete from it all vertices (if any) ranking before $v_{i_0}$ in the topological ordering of $V\setminus F$ (or delete all vertices in common with $V\setminus F$ if $v_{i_0}$ is blank); finally delete all vertices in common with $F\setminus F''$. If the remaining set $R$ is such that $N^+(R)\cap F=F'$ and $R\cap F=F''$, then the triple is associated to the representative $R$, otherwise the triple cannot produce a representative with the required properties. To complete the construction of the vertices of the quotient, we are left to check whether any two $R,R'$ obtained as above are equivalent, which can be easily done by looking at their out-neighborhoods.\\
%\wal{I might becompletely wrong, but let me try. I have the impression that we can simplify... after computing for eac triplet $(F',F",v_{i_0})$ the set $R$ (which is unique as you proved above), I think that we do not need to determine equivalent classes. It does not look "dangerous" to have several nodes for the same equivalent class since the number of nodes is bounded by the upper bounded provide before $4^f(n-f+1)$. To construct the graph, given any valid triplet (and thus given $R$), we compute $N+(R)$ and for each}
Now, to construct an arc $([R],[R'])$ in the quotient, one should make sure there is an arc in $\Pi_k$ from $R$ to a set $R''$ equivalent to $R'$, namely \begin{equation}\label{edgequotient}
 \exists \ R''\subseteq V \mbox{ such that }R'' \subseteq N^+(R), \ \ |N^+(R)\setminus R''|\le k \mbox{ and } N^+(R'')=N^+(R').
 \end{equation}
 This can be done as follows: 
 \begin{enumerate}
     \item compute $S_1:=N^+(N^+(R))$;
     %\item check whether $S_1 \supseteq N^+(R')$: if not, stop without creating the edge;
     \item compute $S_2:=N^-(S_1\setminus N^+(R'))\cap N^+(R)$;
     \item check whether $|S_2|\le k$: if not, stop without creating the arc;
     \item compute $R'':=(N^+(R))\setminus S_2$;
     \item if $N^+(R'')=N^+(R')$, the arc is created, otherwise the arc is not created.
 \end{enumerate}      {Notice the condition of Step 3 is necessary to fulfill Property (\ref{edgequotient}). Indeed all vertices of $S_2$ must be removed from $N^+(R)$ in order to obtain $R''\subseteq N^+(R)$ with $N^+(R'')\subseteq N^+(R')$; moreover the condition of Step 5 is obviously necessary to have Property (\ref{edgequotient}). Finally, notice that, if at Step 5 an arc is created, it means that $R'' \subseteq N^+(R)$ (by construction at Step 4) and $|N^+(R)\setminus R''|=|S_2|\le k$ (after Step 3) and $N^+(R'')=N^+(R')$ (as checked at Step 5), therefore $R''$ fulfills Property (\ref{edgequotient}).}\\ 
 To summarize, our FPT algorithm will start by finding a feedback vertex set $F$ of minimum size; then compute the topological order of $V\setminus F$; then compute the quotient digraph as described above and search there for a path from $[V]$ to $[\emptyset]$. Notice that $F$ can be calculated, for example, using an algorithm proposed by Dom et al. \cite{DOM201076} which runs in time $O(2^fn^3)$; the time complexity of all other routines is dominated by that of constructing the arcs of the quotient digraph: here for every possible pair of vertices of the quotient, steps 1 to 5 are executed (in $O(n^2)$ time), so the arc creation procedure takes time $O((4^fn)^2n^2)$. It follows that the time complexity of our FPT algorithm is $O(2^{4f}n^4)$. 
\end{proof}

\section{Connections with  binary matrix mortality}
\label{sec:morta}
{
Given a finite set of integer valued square matrices, the mortality problem asks whether they are mortal, namely whether there exists a product of these matrices that gives the zero matrix. This problem is undecidable even for $3x3$ matrices \cite{paterson1970}. We focus on matrices having nonnegative entries: in this case, what only matters for mortality is whether or not an entry is $0$, so the problem could be thought over binary matrices. Nonnegative (or binary) matrix mortality is proven to be equivalent to the problem of deciding whether a nondeterministic finite state automaton has a killing word \cite{doi:10.1137/19M1250893} and this problem is PSPACE-complete \cite{KAO20095010}. We will see how the problem of deciding whether a digraph has cop number at most $k$ reduces to a binary matrix mortality problem.
}

Let $B$ be the adjacency matrix of $D=(V,A)$ and let $C$ be the adjacency matrix of a complete  graph (square matrix with $1$ everywhere). 
Given  $W \subset V$, let $\bar{I}_{W}$ be the diagonal matrix where $\bar{I}_{i i} = 0$, $\forall i \in W$ and $\bar{I}_{i i} = 1$ otherwise.  Observe that $B  \bar{I}_{W}$ is then the matrix obtained from $B$ by switching to $0$ all coefficients of each column related to a vertex $i \in W$.  $B \bar{I}_{W}$ represents the possible moves that the robber can have if cops are on $W$ without being immediately captured.  
Assume that $W_1, W_2, ..., W_l$ is a winning strategy. Initially, the robber can choose any position outside $W_1$. This can be seen as a robber move in a complete graph where cops are on $W_1$ and can then be represented through the matrix $C \bar{I}_{W_1}$. For the second time step, the possible robber moves are modeled by matrix $B \bar{I}_{W_2}$. Assuming that $W_1, W_2, ..., W_l$ is a winning strategy is then simply equivalent to say that $C \bar{I}_{W_1} B \bar{I}_{W_2} ...B \bar{I}_{W_l} = 0$ (the existence of any nonzero coefficient of this matrix immediately implies the existence of a robber strategy to escape). Since  $C$ is the all-one matrix and all matrices we are manipulating are non-negative (in fact they are binary), the previous equality holds, if and only if, $\bar{I}_{W_1} B \bar{I}_{W_2} ...B \bar{I}_{W_l} = 0$, so matrix $C$ can be skipped. This is stated below.
\begin{lemma}
  $W_1, W_2, ..., W_l$ is a winning strategy, if and only if,  $\bar{I}_{W_1} B \bar{I}_{W_2} ...B \bar{I}_{W_l} = 0$.   
  \label{lem:matr}
\end{lemma}

An immediate corollary is that $D$ and its reverse digraph (say $\overleftarrow{D}$) have the same cop number and the same capture time.  

\begin{corollary}
    $\mbox{cn}(D) = \mbox{cn}(\overleftarrow{D})$ and $ct(D) = ct(\overleftarrow{D})$
    \label{coro:rev}
\end{corollary}
\begin{proof}
Let $W_1, W_2, ..., W_l$ be a winning strategy over $D$. From Lemma \ref{lem:matr}, we can write that $\bar{I}_{W_1} B \bar{I}_{W_2} ...B \bar{I}_{W_l} = 0$. Taking the transpose and using the fact that matrices of type  $\bar{I}_{W}$ are symmetric, we deduce that $\bar{I}_{W_l} B^T \bar{I}_{W_{l-1}} ...B^T \bar{I}_{W_1} = 0$. Since the adjacency matrix of $\overleftarrow{D}$ is $B^T$ and using again Lemma \ref{lem:matr}, we can assert  that $W_{l}, W_{l-1}, ..., W_1$ is a winning strategy over $\overleftarrow{D}$.
\end{proof}

Corollary \ref{coro:rev}  implies that any valid bound for $\mbox{cn}(D) = \mbox{cn}(\overleftarrow{D})$. The lower bound recalled in Introduction and generalized in  Lemma \ref{lem:obvious} can then be improved into $\max_{S \subset V} \max\left(\delta^+(D[S]),\delta^-(D[S]) \right)$, where $\delta^-(D[S])$ is the minimum indegree in the induced graph $D[S]$.
\begin{corollary}
    $\mbox{cn}(D) \ge \max\limits_{S \subset V} \mbox{ }  \max \left(\delta^+(D[S]),\delta^-(D[S]) \right)$.
\end{corollary}

An obvious observation is that the cop number of a graph does not change if a vertex of out-degree $0$ is deleted. From Corollary \ref{coro:rev}, we can assert that it also stays the same if a vertex of in-degree $0$ is removed.  Let us then assume that $\delta^+(D) \geq 1$ and $\delta^-(D) \geq 1$. This is equivalent to say that each column and each row of the adjacency matrix $B$ contains at least one nonzero coefficient.

\begin{lemma}
    Let $D =(V,A)$ be such that $\delta^+(D) \geq 1$ and $\delta^-(D) \geq 1$, and let $W_1, W_2, ..., W_l$  represent a cop strategy. The three following properties are then equivalent:
    \begin{enumerate}
        \item $\bar{I}_{W_1} B \bar{I}_{W_2} ...B \bar{I}_{W_l} = 0$.
    \item $B \bar{I}_{W_1} B \bar{I}_{W_2} ...B \bar{I}_{W_l} = 0$.
        \item $\bar{I}_{W_1} B \bar{I}_{W_2} ...B \bar{I}_{W_l} B = 0$.
    \end{enumerate}
    \end{lemma}
\begin{proof}
Property 1 obviously implies properties 2 and 3 since we are just multiplying a $0$ matrix by $B$ (either left or right). Suppose now that 2 holds: $B \bar{I}_{W_1} B \bar{I}_{W_2} ...B \bar{I}_{W_l} = 0$. Assume by contradiction that the matrix $K = \bar{I}_{W_1} B \bar{I}_{W_2} ...B \bar{I}_{W_l}  \neq 0$ implying that $K_{ij} = 1$ for some $i, j$. Since the $i^{th}$ column of $B$ contains at least one nonzero coefficient, $B K$ cannot be $0$ contradicting property 2. In other words, 2 implies 1. One can easily show that $3$ implies $1$ using the fact that each row of $B$ contains at least one nonzero coefficient. 
\end{proof}

The previous lemma allows us to make a connection with the binary matrix mortality problem. %Roughly speaking, given a set of $m$ square binary matrices $B_1$, $B_2$,...,$B_m$,  we ask whether the zero matrix can be expressed as a finite product of matrices from this set.*****\wal{some literature here ***}.
Assume, for example, that we aim to check whether $\mbox{cn}(D) \le k$. Then we can build for each subset $W$ of size $k$ the matrix $B_W = B \bar{I}_{W}$. From the previous lemma (property 2), $\mbox{cn}(D) \le k$ holds, if and only if,  some product $B_{W_1} B_{W_2}...B_{W_l}$ is equal to $0$ where $|W_i|=k$.   

Remember that we proved that even deciding whether $\mbox{cn}(D) \le 1$ is NP-hard. In this case, each matrix $B_{W}$ correspond to the matrix $B$ with one column switched to $0$, and the number of such matrices is equal to $n$. In other words, we can say that the binary matrix mortality problem remains NP-hard even if the number of matrices is $n$ and each matrix is obtained from the same main $n$-square matrix by setting to $0$ the coefficients of a column. On the contrary, in case the above $n$-square matrix is symmetric and has a $0$ diagonal, checking mortality is polynomial, because it reduces to checking whether an undirected loopless graph is $1$-copwin and this can be polynomially characterized \cite{HASLEGRAVE201412}.
}

\section{Conclusion and future work}
{
We have studied a cops and robber game on directed graphs, which extends the hunter and rabbit problem defined on undirected graphs. We mostly focused on complexity issues regarding the computation of the cop number of a digraph and showed it is NP-hard to decide whether a digraph is $1$-copwin. We proved the latter problem is polynomial on tournaments by showing that a tournament is $1$-copwin if and only if it has a feedback vertex set of size $1$. We also provided an FPT algorithm parameterized by the size of a minimum feedback vertex set. It would be natural and interesting to assess the complexity of determining the cop number of a tournament.\\
It would also be nice to establish whether {deciding if the cop number equals one} on general digraphs is in NP. The problem might well be PSPACE-complete as the more general problem of binary matrix mortality. Another, related, possible future direction is determining tighter bounds for the maximum capture time over digraphs with $n$ vertices. We have seen that an upper bound of $2^n -1$ exists. If a polynomial upper bound should be found, this would imply computing the cop number is in NP. On the other hand, in \cite{10.1007/978-3-031-72621-7_8} is shown a lower bound of $2^n-1$ for the length of the shortest mortal product of $n$ binary $n\times n$ matrices and it might somehow be extended to the case of the maximum capture time.}
{Other future directions include further study of the approximability of the cop number and, possibly, the design of approximation algorithms.}

\section*{Acknowledgments}
We would like  to thank Antoine Amarilli, Harmender Gahlawat and Vlad-Stefan Vergelea for pointing out that the NP-hardness proofs of Theorems 4.1 and 4.2, initially published in \cite{ourpaper}, required corrections.\\
This research benefited from the support of the FMJH Program Gaspard Monge for optimization and operations research and their interactions with data science.
\bibliographystyle{abbrvnat}
\bibliography{sample}

\end{document}